\documentclass{article} 
\usepackage[centertags,sumlimits,intlimits,namelimits,reqno]{amsmath}
\usepackage{latexsym,amsfonts,amssymb,exscale,enumerate,amsthm}
\usepackage[colorlinks,linkcolor=blue,citecolor=red,urlcolor=blue]{hyperref}

\usepackage{tikz}

        \newcommand\N{{\mathbb N}}   
           
        \newcommand\R{{\mathbb R}}   
        \newcommand\maps{{\colon}}   
        \newcommand{\define}[1]{{\bf \boldmath #1}}

	\newtheorem{theorem}{Theorem}

	\newtheorem{definition}{Definition}

       \newtheorem*{proposition*}{Proposition}

\theoremstyle{remark}

\newtheorem*{example*}{Example}
       
\begin{document}   
 
	\begin{center}   
	{\bf  Quantum Techniques for Studying Equilibrium \\
in Reaction Networks \\}   
        \vspace{0.3cm}
	{\em John\ C.\ Baez \\}
        \vspace{0.3cm}
	{\small 
        Department of Mathematics \\
        University of California \\
        Riverside CA 92521, USA \\
        and \\
        Centre for Quantum Technologies  \\
        National University of Singapore \\
        Singapore 117543 \\
 } 
        \vspace{0.3cm}
	{\em Brendan Fong \\}
        \vspace{0.3cm}
        {\small Department of Computer Science  \\
        University of Oxford  \\
        United Kingdom OX1 3QD  \\ }
	\vspace{0.3cm}   
        {\small email:  baez@math.ucr.edu, brendan.fong@cs.ox.ac.uk\\} 
	\vspace{0.3cm}   
	{\small June 11, 2013}
	\vspace{0.3cm}   
	\end{center}

\begin{abstract}
\noindent Anderson, Craciun, and Kurtz have proved that a
stochastically modelled chemical reaction system with mass-action
kinetics admits a stationary distribution when the deterministic model
of the same system with mass-action kinetics admits an equilibrium
solution obeying a certain `complex balance' condition.  Here we
present a proof of their theorem using tools from the theory of
second quantization: Fock space, annihilation and creation operators,
and coherent states.  This is an example of `stochastic mechanics', 
where we take techniques from quantum mechanics and replace 
amplitudes by probabilities. We explain how the systems studied here can
be described using either of two equivalent formalisms: reaction networks, 
as in chemistry, or stochastic Petri nets, as in some other disciplines.
\end{abstract}

\section{Introduction}

Systems with uncertainty are often modelled in two ways. The first of
these is via the deterministic evolution of the expected state of the system. 
The second, more detailed approach describes the time evolution of a 
probability distribution or quantum state.  Often the second approach
reduces to the first in a suitable limit.

One prominent class of systems modelled in both ways are field theories, 
which come in both classical and quantum versions.   Another class of such 
systems are chemical reactions.  Chemists describe the time evolution of
chemical concentrations in a deterministic way using \emph{rate equations},
but they also describe interactions between chemical species stochastically,
using \emph{master equations}.  

In this paper we aim to demonstrate the analogy between quantum field
theory and chemistry, casting our discussion of chemical master equations 
in the language of Fock space and creation and annihilation operators.   
This work is part of a broader program emphasizing the links between quantum
mechanics and a subject we call `stochastic mechanics', where probabilities
replace amplitudes \cite{BB,BF,B}.

Using second quantization methods to model classical systems goes back
to Doi \cite{D}, and was developed further by Grassberger and
Scheunert \cite{GS}, among others.  However, there is an extensive
line of work in mathematical chemistry that does not use these
techniques, but might profit from them.  To illustrate how they work,
we use them here to provide a new, compact proof of a theorem by
Anderson, Craciun, and Kurtz \cite{ACK}.  This theorem says that under
certain assumptions, equilibrium solutions of deterministic evolution
equations give equilibrium solutions of stochastic ones.

Chemists specify both the rate equation and the master equation
starting from a `reaction network'.  The language of Petri nets is
equivalent to that of reaction networks, but more widely used beyond
discussions of chemistry.  In particular, stochastic Petri nets are
used in systems biology, epidemiology, and ecology \cite{GP,K,KRS,
Wil06}.  Unfortunately, results proved in one language are not
instantly accessible to researchers only familiar with the other.
Thus, we start by explaining how to translate between reaction
networks and Petri nets.  Keeping this in mind, it becomes clear that
the Anderson--Craciun--Kurtz theorem is not merely a result about
chemical reactions: it applies to a wide class of stochastic systems.

\section{Reaction networks and Petri nets}

Reaction networks model the salient chemical interactions in an
isolated chemical system. These networks, first defined by Aris
\cite{A}, consist of a set of chemical species together with a set of
reactions that may take place among them.  More precisely:

\begin{definition}
A \define{reaction network} $(S,K,T,s,t)$ consists of:
\begin{itemize}
\item a finite set $S$,
\item a finite set $K$ of multisets with elements in $S$,
\item a finite set $T$,
\item functions $s,t \maps T \to K$.
\end{itemize}
We call the elements of $S$ \define{species}, those of $K$
\define{complexes}, and those of $T$ \define{reactions}.  We call
$s(\tau)$ and $t(\tau)$ the \define{reactants} and \define{products}
of the reaction $\tau \in T$ respectively.
\end{definition}

\noindent
We should note that chemists use the term `chemical reaction network'.
While chemists call the elements of $T$ `reactions', the Petri net
literature uses the term `transitions' for essentially the same thing.
We often write complexes as linear combinations of elements of $S$
with natural number coefficients, writing for example the multiset
$\{A,A,B\}$ as $2A+B$.  If a reaction $\tau \in T$ has the complex
$\kappa$ as its reactants and $\kappa'$ as its products, we write
$\tau \maps \kappa \to \kappa'$.

Reaction networks admit a useful and easy-to-read representation as
directed graphs, with $K$ as the set of vertices, and a directed edge
from the complex $a \in K$ to the complex $b \in K$ for each reaction
with $a$ as its reactants and $b$ as its products. For example,
consider the reaction network with
\[ S = \big\{A,B,C,AC\big\} \]
as its set of species,
\[ K = \big\{\varnothing, A, B,AC,A+C,2B+C\big\} \]
as its set of complexes, and with four reactions
\[   R = \{ \alpha, \beta, \gamma, \delta \}  \]
such that
\[
\begin{array}{l}    \alpha  \maps \varnothing \to A  \\
\beta \maps  B \to \varnothing  \\
\gamma \maps A + C \to AC \\
\delta \maps AC \to 2B + C  
\end{array}
\]
This reaction network describes the transformation of one molecule of
species $A$ into two of species $B$ with the help of the catalyst $C$,
in a system where $A$ is replenished and $B$ is consumed.  It may be
pictured as this graph:
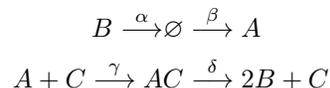
\begin{align*}
B \stackrel{\alpha}{\longrightarrow} 
&\varnothing \stackrel{\beta}{\longrightarrow} A \\
A+C \stackrel{\gamma}{\longrightarrow} A 
& C \stackrel{\delta}{\longrightarrow} 2B+C 
\end{align*}

In general, the graph conveys all the information in the reaction
network, except when there are species that do not appear in any
complex.  The graph may have multiple edges or self-loops.  It is thus
the kind of graph sometimes called a `directed multigraph' or
`quiver'.  These terms may obscure the simplicity of the concept, so
for the purposes of this paper we simply say:

\begin{definition} A \define{graph} consists of a set $E$ of 
\define{edges}, a set $V$ of \define{vertices}, and \define{source}
and \define{target} maps $s,t \maps E \to V$.
\end{definition}

\noindent A reaction network is thus equivalent to a finite set $S$ of
species, a finite set $K$ of multisubsets of $S$, and a graph with $K$
as its set of vertices and some finite set $R$ of edges.

An equivalent framework for describing chemical processes, now used
more generally to describe concurrent processes across a wide variety
of sciences, is that of Petri nets \cite{PR}.

\begin{definition}
A \define{Petri net} $(S,T,i,o)$ consists of:
\begin{itemize}
\item a finite set $S$,
\item a finite set $T$,
\item functions $i,o: S \times T \to \N$.
\end{itemize}
We call the elements
of $S$ \define{species} and those of $T$ \define{transitions}.  We 
think of $i$ and $o$ as specifying the number of each species appearing 
as the \define{inputs} and \define{outputs} respectively to each transition. 
\end{definition}

\noindent
In the Petri net literature the species are often called `places'.
For convenience we shall often write $k=|S|$ for the number of species
present in a Petri net, and identify the set $S$ with the set $\{1,
\dots, k\}$.

Each reaction network may be considered as equivalent to some Petri
net.  Here the set of species of the Petri net is the same as that of
the equivalent reaction network.  The set of transitions of the Petri
net corresponds to the set of reactions of the reaction network.  The
input and output functions of each transition respectively specify the
number of times each species appears as a reactant and product of the
corresponding reaction. For example, the above example of a reaction
network is equivalent to the Petri net with
\begin{align*}
S &= \big\{A,B,C,AC\big\}, \\
T &= \big\{\alpha, \beta, \gamma, \delta \big\},
\end{align*}
the input function $i \maps S \times T \to \N$ taking value $1$ on the pairs 
$(A,\beta)$, $(C,\beta)$, $(AC,\gamma)$ and value $0$ otherwise, and the 
output function $o \maps S \times T \to \N$ taking value $1$ on $(A,\alpha)$ 
and $(AC,\beta)$, value $2$ on $(B,\gamma)$, and value $0$ otherwise.

Petri nets are also commonly represented as graphs, although not in
the same way as their equivalent reaction networks. In this case we
draw a graph whose set of vertices is $S \cup T$, with precisely
$i(j,\tau)$ directed edges drawn from the species $j \in S$ to the
transition $\tau \in T$, and precisely $o(j,\tau)$ directed edges
drawn from the transition $\tau \in T$ to the species $j \in S$.  So,
this graph is `bipartite': there are no edges going between two
vertices that both represent states, and no edges going between two
vertices that both represent transitions.

Our example Petri net is represented by this graph: 
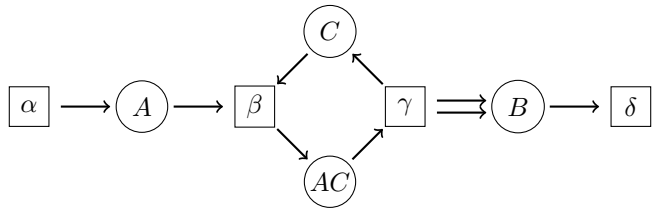
\[
\begin{tikzpicture}
\node [shape=rectangle,minimum height=15pt,minimum width=15pt,draw=black] at (-4,0) {$\alpha$};
\node [shape=circle,draw=black] at (-2.5,0) {$A$};
\node [shape=rectangle,minimum height=15pt,minimum width=15pt,draw=black] at (-1,0) {$\beta$};
\node [shape=circle,draw=black] at (0,1) {$C$};
\node [shape=circle,inner sep = 1pt, draw=black] at (0,-1) {$AC$};
\node [shape=rectangle,minimum height=15pt,minimum width=15pt,draw=black] at (1,0) {$\gamma$};
\node [shape=circle,draw=black] at (2.5,0) {$B$};
\node [shape=rectangle,minimum height=15pt,minimum width=15pt,draw=black] at (4,0) {$\delta$};
\draw [->, thick,shorten >=12pt, shorten <=12pt]  (-4,0) to (-2.5,0); 
\draw [->, thick,shorten >=12pt, shorten <=12pt]  (-2.5,0) to (-1,0); 
\draw [->, thick,shorten >=12pt, shorten <=12pt]  (0,1) to (-1,0); 
\draw [->, thick,shorten >=12pt, shorten <=12pt]  (-1,0) to (0,-1); 
\draw [->, thick,shorten >=12pt, shorten <=12pt]  (0,-1) to (1,0); 
\draw [->, thick,shorten >=12pt, shorten <=12pt]  (1,0) to (0,1); 
\draw [->, thick,shorten >=12pt, shorten <=12pt]  (1,0.08) to (2.5,0.08); 
\draw [->, thick,shorten >=12pt, shorten <=12pt]  (1,-0.08) to (2.5,-0.08); 
\draw [->, thick,shorten >=12pt, shorten <=12pt]  (2.5,0) to (4,0); 
\end{tikzpicture}
\]
Note that there are two edges from $\gamma$ to $B$, as the transition
$\gamma$ produces two instances of species $B$.

\section{Stochastic Petri nets with mass-action kinetics}

It is often useful to talk of the amount of each chemical present in a
chemical system at a given time. The possible amounts are represented
by what are termed the `states' of a Petri net.  By a \define{pure
state} of a Petri net we refer to a vector $x \in \N^{k}$ with $i$th
entry $x_i$ specifying the number of instances of the $i$th species.
This generalises in two ways.  First, we define a \define{classical
state} to be any vector $x \in [0,\infty)^{k}$ of nonnegative real
numbers, and think of such a state as specifying an expected number or
concentration of the species present. Second, we define a
\define{mixed state} to be a probability distribution $\psi$ over the
pure states of the Petri net. Of great interest is how states---of
both generalized types---evolve or change over time.

In modelling the evolution of the state of a Petri net, we must first
specify a kinetics for the system. In this paper we consider systems
that evolve by the law of mass action \cite{HJ}. Very roughly, this
law states that the rate at which a transition occurs is proportional
to the product of the concentrations of all its input species.  We
call the constant of proportionality for each transition its `rate
constant'. A Petri net where each transition has a rate constant is
called a `stochastic Petri net':

\begin{definition}
A \define{stochastic Petri net} $(S,T,i,o,r)$ consists of a Petri net
$(S,T,i,o)$ together with a function $r \maps T \to (0,\infty)$
specifying a \define{rate constant} for each transition.
\end{definition}

\noindent
Just as Petri nets are equivalent to reaction networks, a stochastic Petri
net is equivalent to a \define{stochastic reaction network}, meaning 
a reaction network together with a rate constant for each reaction.

A stochastic Petri net gives rise to an evolution principle for each
of our generalized notions of state above. The `rate equation'
describes the system with continuous parameters changing in a
deterministic way, saying how a classical state changes with time. On
the other hand, the `master equation' describes the system with a
mixed state---a probability distribution over the pure states of the
system---and treats the evolution of pure states as a stochastic
process, saying how this probability distribution changes with time.

The master equation is more fundamental.  It plays a role similar to
the equations of quantum electrodynamics, which describe the
amplitudes for creating and annihilating individual photons. The rate
equation is less fundamental: it resembles the classical Maxwell
equations, which describe changes in the electromagnetic field in a
deterministic way. The classical Maxwell equations are an
approximation to quantum electrodynamics which becomes good
in a certain `classical limit' where there are many photons 
in each state.  Similarly, although the rate equation and
master equation describe distinct notions of evolution, under certain
conditions the rate equation can be derived from the master equation
in the limit where the number of things of each species become large,
and the uncertainty in these numbers becomes negligible \cite{AK,B}. In
this case taking the expected value of an equilibrium solution for
the master equation gives an equilibrium solution of the rate
equation.

The Anderson--Craciun--Kurtz theorem shows that in some cases the
converse holds too: from an equilibrium solution of the
rate equation we can sometimes construct an equilibrium solution of
the master equation.  Remarkably, this solution is precise rather than
approximate---and it is given by a simple formula.  

\subsection*{The rate equation}

Let $(S,T,i,o,r)$ be a stochastic Petri net with $k = |S|$
species. For each species $j \in S$ and transition $\tau \in T$, 
write $s_j(\tau) = i(j,\tau)$ for the number of times the species $j$
appears as an input to the transition $\tau$, and similarly write
$t_j(\tau) = o(j,\tau)$ for the number of times it appears an output.
We thus have vectors
\[   s(\tau) = (s_1(\tau), \dots , s_k(\tau)) \in \N^k \]
and
\[   t(\tau) =  (t_1(\tau), \dots , t_k(\tau)) \in \N^k \]
describing the reactants and products of the transition $\tau$.

The rate equation models the time evolution of a classical state of a
stochastic Petri net in a deterministic way. According to the
deterministic law of mass action, for each transition $\tau$ the rate
of change of $x$ is the product of:
\begin{itemize}
\item the change $t(\tau) - s(\tau)$
in numbers of each species due to $\tau$;
\item the concentration of each input species $i$ of $\tau$ raised to
the power given by the number of times it appears as an input, namely
$s_i(\tau)$;
\item the rate constant $r(\tau)$ of $\tau$.
\end{itemize}

\begin{definition} 
The \define{rate equation} for a stochastic Petri net $(S,T,i,o,r)$ is
\[
\frac{d}{d t}x(t) = \sum_{\tau \in T} r(\tau) (t(\tau) - 
s(\tau)) x(t)^{s(\tau)},
\]
where $x \maps \R \to [0,\infty)^k$ and we have
have used multi-index notation to define
\[  x^{s(\tau)} = x_1^{s_1(\tau)} \cdots x_k^{s_k(\tau)}. \]
\end{definition}

We will in particular concern ourselves with states $c \in
[0,\infty)^k$ for which $x(t) = c$ for all $t \in \R$ is a solution of
the rate equation. In such a state the rate equation says that the net
rate of production of each species is zero. We call such a state an
\define{equilibrium solution} of the rate equation.

\subsection*{Second quantization and the master equation}

Modelling the evolution of a mixed state is a bit more involved, and
here methods from quantum mechanics provide an interesting
perspective. Following second quantization techniques used in quantum
field theory, we begin our discussion by introducing a power series
notation to describe the mixed states of a stochastic Petri net, and
then creation and annihilation operators to describe how these mixed
states may evolve. This will allow us to give a compact expression of
the master equation.

Fix some mixed state. For each $n \in \N^k$, let $\psi_n$ be the
probability that for each $i$ we have exactly $n_i \in \N$ of the
$i$th species present in this mixed state. We then express this mixed
state as a formal power series in some formal variables $z_1, \dots,
z_k$, with the coefficient of 
\[    z^n = z_1^{n_1} \cdots z_k^{n_k} \]
being the probability $\psi_n$.  That is, we write our mixed state as
$\Psi = \sum_{n \in \N^k} \psi_n z^n$. Note that as $\Psi$
represents a mixed state, and a mixed state is a probability
distribution, the coefficients $\psi_n$ of must sum to 1. Indeed, the
mixed states are precisely the formal power series $\Psi$ with
coefficients summing to 1. The simplest examples of states are the
monomials $z^n = z_1^{n_1} \cdots z_k^{n_k}$; these are the pure
states.

The simplest changes to the state of a Petri net are the addition, or
creation, of a new instance of the $i$th species, and the removal, or
annihilation, of an existing instance of the $i$th species. Moreover, all
changes of state can be seen as sequences of such operations. To
provide notation for processes on mixed states, we thus define
creation and annihilation operators on formal power series. These
definitions are inspired by their analogues in quantum field
theory. Our notation here follows that used there directly, where the
creation and annihilation operators may be seen as adjoint linear
operators.

Let $1 \le i \le k$. The \define{creation operator} $a_i^\dagger$ of the
$i$th species is given by
\[
a_i^\dagger \Psi = z_i \Psi.
\]
This takes a pure state $z^n = z_1^{n_1} \cdots z_i^{n_i}\cdots
z_k^{n_k}$ to the pure state $z_1^{n_1} \cdots z_i^{n_i+1} \cdots
z_k^{n_k}$ with one additional instance of the $i$th species. The
corresponding \define{annihilation operator} is given by formal
differentiation:
\[
a_i \Psi = \frac{\partial}{\partial z_i} \Psi.
\]
Given a pure state $z^n$, this represents moving to the state with one
fewer instance of the $i$th species, albeit with a scaling coefficient
$n_i$. This may be interpreted combinatorially, representing the fact
that there are $n_i$ ways to annihilate one of the $n_i$ instances of
the $i$th species present. This is appropriate as in this setting the
law of mass action states that the rate of a transition is
proportional to the number of distinct subsets of the present species
that can form an input to the transition.

We now have the tools required for our definition of the
master equation:

\begin{definition}
The \define{master equation} for a stochastic Petri net $(S,T,i,o,r)$ is
\[
\frac{d}{d t} \Psi(t) = H \Psi(t),
\]
where the \define{Hamiltonian} $H$ is given by
\[
H = \sum_{\tau \in T} r(\tau) \, 
\left({a^\dagger}^{t(\tau)} - {a^\dagger}^{s(\tau)}\right) \, a^{s(\tau)}.
\]
Here again we use multi-index notation to define, for any $n \in
\N^k$,
\[   a^n = a_1^{n_1} \cdots a_k^{n_k} \]
and
\[   {a^\dagger}^n = {a_1^\dagger}^{n_1} \cdots {a_k^\dagger}^{n_k} .\]
\end{definition}

For each transition $\tau$, the first term ${a^\dagger}^{t(\tau)}
a^{s(\tau)}$ of the Hamiltonian describes how $s_i(\tau)$ things of
the $i$th species get annihilated, and $t_i(\tau)$ things of the $i$th
species get created. The second term $- {a^\dagger}^{s(\tau)}
a^{s(\tau)}$ is a bit harder to understand, but it describes how the
probability that nothing happens---that we remain in the same pure
state---decreases as time passes. This term must take precisely the
form it does to ensure conservation of total probability.  For a more
detailed derivation of the Hamiltonian $H$, see the companion paper
\cite{B}.

\section{The Anderson--Craciun--Kurtz Theorem}

From a certain class of equilibrium solutions of the rate equation,
Anderson, Craciun, and Kurtz have proved that we may obtain 
equilibrium solutions of the master equation \cite{ACK}. We now
translate their proof into the language of annihilation and creation
operators, emphasizing the analogies between our Petri net systems and
ideas in quantum mechanics.  In particular, their equilibrium solution
of the master equation is what people call a `coherent state' in quantum 
mechanics.

To construct an equilibrium solution of the master equation from one
for the rate equation, we require a `complex balanced' solution. Here
the term `complex' again refers to a collection of instances of the
given species, and a complex balanced solution is a solution in which
not only the net rate of production of each species is zero, but
furthermore the net rate of production of each possible
\emph{collection} of species is zero. For our Petri net, complexes are
elements of $\N^k$, and the complexes of particular interest are the
input complex $s(\tau)$ and the output complex $t(\tau)$ of each
transition $\tau$.

\begin{definition}
We say a classical state $c \in [0,\infty)^k$ is \define{complex
balanced} if for all complexes $\kappa \in \N^k$ we have
\[
\sum_{\{\tau : s(\tau) = \kappa\}} r(\tau) c^{s(\tau)} =
\sum_{\{\tau : t(\tau) = \kappa\}} r(\tau) c^{s(\tau)}.  
\label{complex_balanced}
\]
\end{definition}
The left hand side of the above equation, which sums over the
transitions in which the input complex is $\kappa$, gives the rate of
consumption of the complex $\kappa$.  The right hand side, which sums
over the transitions in which the output complex is $\kappa$, gives
the rate of production of $\kappa$.  Thus the above equation requires
that the net rate of production of the complex $\kappa$ is zero when
the number of things of each species is given by the vector $c$.

Any complex balanced classical state gives an equilibrium solution of
the rate equation. Indeed, if $c \in [0,\infty)^k$ is a complex
balanced classical state, then if $x(t) = c$ for all $t$ we have
\begin{multline*}
\frac{d}{d t}x(t) = 
\sum_{\tau \in T} r(\tau) (t(\tau) - s(\tau)) c^{s(\tau)}  =
\sum_{\kappa \in \N^k}\kappa
\left(\sum_{\{\tau:t(\tau) = \kappa\}} r(\tau) c^{s(\tau)} 
- \sum_{\{\tau:s(\tau) = \kappa\}} r(\tau)c^{s(\tau)}\right)
=0,
\end{multline*}
and so $x(t) = c$ is an equilibrium solution.  In general, not every
equilibrium solution of the rate equation is complex balanced.
However, the `deficiency zero theorem' \cite{F,HJ} says this holds for
a large class of stochastic Petri nets.

The Anderson--Craciun--Kurtz theorem says that any complex balanced
classical state also gives rise to an equilibrium solution of the
master equation.  Such solutions have a special form: they are
`coherent states'. For any $c \in [0,\infty)^k$, we define the
\define{coherent state} with expected value $c$ to be the mixed state
\[
\Psi_c = \displaystyle{\frac{e^{c \cdot z}}{e^c}} 
\]
where $c \cdot z$ is the dot product of the vectors $c$ and $z$, and
we set $e^c = e^{c_1+\dots+c_n}$.  Equivalently,
\[  \Psi_c = \frac{1}{e^c} \sum_{n \in \N^k} \frac{c^n}{n!}z^n, \] 
where $c^n$ and $z^n$ are defined as products in our usual way, and
$n! = n_1! \, \cdots \, n_k!$. The name `coherent state' comes from
quantum mechanics \cite{KS}, in which we think of the coherent state
$\Psi_c$ as the quantum state that best approximates the classical
state $c$. In the state $\Psi_c$, the probability of having $n_i$
things of the $i$th species is equal to
\[                e^{-c_i} \, \displaystyle{\frac{ c_i^{n_i}}{n_i!} }. \] 
This is precisely the definition of a Poisson
distribution with mean equal to $c_i$. The state $\Psi_c$ is thus a
product of independent Poisson distributions.

It is remarkable that such a simple state can give an equilibrium
solution of the master equation, even for very complicated stochastic
Petri nets.  But it is nonetheless true when $c$ is complex balanced:

\begin{theorem}[{\bf Anderson--Craciun--Kurtz}]
Suppose $c \in [0,\infty)^k$ is a complex balanced classical state of
a stochastic Petri net, and suppose that $H$ is the Hamiltonian for
the master equation of this stochastic Petri net. Then $H \Psi_c = 0$.
\end{theorem}

\begin{proof}
Since $\Psi_c$ is a constant multiple of $e^{c \cdot z}$, it suffices to
show $H e^{c \cdot z} = 0$. Recall that
\[
H = \sum_{\tau \in T} r(\tau) 
\left( {a^\dagger}^{t(\tau)} -{a^\dagger}^{s(\tau)} \right) \, a^{s(\tau)}.
\]
As the annihilation operator $a_i$ is given by differentiation with
respect to $z_i$, while the creation operator $a^\dagger_i$ is just
multiplying by $z_i$, we have 
\[   a^{s(\tau)} e^{c \cdot z} = c^{s(\tau)} e^{c \cdot z},  \qquad 
 {a^\dagger}^{s(\tau)} e^{c \cdot z} = z^{s(\tau)} e^{c \cdot z}, \]
and thus
\[
H e^{c \cdot z} =
 \sum_{\tau \in T} r(\tau) \, c^{s(\tau)} \left( z^{t(\tau)} 
- z^{s(\tau)} \right) e^{c \cdot z} = 
\sum_{i \in \N^k} \sum_{\tau \in T} r(\tau)c^{s(\tau)}\left(z^{t(\tau)}
\frac{c^i}{i!}z^i - z^{s(\tau)}\frac{c^i}{i!}z^i\right),
\]
where the second equality is given by expanding $e^{c \cdot z}$. If we then
define negative powers to be zero and reindex, we may write
\[
H e^{c \cdot z}  = 
\sum_{i \in \N^k} \sum_{\tau \in T} 
r(\tau)c^{s(\tau)}\left(\frac{c^{i-t(\tau)}}{(i-t(\tau))!}z^i - 
\frac{c^{i-s(\tau)}}{(i-s(\tau))!}z^i\right).
\]

Hence, to show $H e^{c \cdot z} = 0$, it is enough to show
\[
\sum_{i \in \N^k} \sum_{\tau \in T} r(\tau)c^{s(\tau)}
\frac{c^{i-t(\tau)}}{(i-t(\tau))!}z^i =\sum_{i \in \N^k} 
\sum_{\tau \in T} r(\tau)c^{s(\tau)}\frac{c^{i-s(\tau)}}{(i-s(\tau))!}z^i.
\]
We do this by splitting the sum over the transitions $T$ according to
output and then input complexes, making use of the complex balanced
condition in the equality marked ($\ast$):
\begin{align*}
&\quad \sum_{i \in \N^k} \sum_{\tau \in T} 
r(\tau)c^{s(\tau)}\frac{c^{i-t(\tau)}}{(i-t(\tau))!}z^i 
\\
&=\sum_{i \in \N^k} \sum_{\kappa \in \N^k} 
\sum_{\{\tau : t(\tau)=\kappa\}}  
r(\tau)c^{s(\tau)}\frac{c^{i-t(\tau)}}{(i-t(\tau))!}  \, z^i 
\\
&= \sum_{i \in \N^k} \sum_{\kappa \in \N^k} 
\frac{c^{i-\kappa}}{(i-\kappa)!}z^i \sum_{\{\tau : t(\tau) = \kappa\}}
r(\tau)c^{s(\tau)} 
\\
&= \sum_{i \in \N^k} \sum_{\kappa \in \N^k} 
\frac{c^{i-\kappa}}{(i-\kappa)!}z^i \sum_{\{\tau : s(\tau) = \kappa\}} 
r(\tau)c^{s(\tau)} \tag{$\ast$}\\
&= \sum_{i \in \N^k} \sum_{\kappa \in \N^k} 
\sum_{\{\tau : s(\tau) = \kappa\}}  
r(\tau)c^{s(\tau)}\frac{c^{i-s(\tau)}}{(i-s(\tau))!} \, z^i \\
&=\sum_{i \in \N^k} \sum_{\tau \in T} r(\tau)
c^{s(\tau)}\frac{c^{i-s(\tau)}}{(i-s(\tau))!}z^i.
\end{align*}
This proves the theorem.
\end{proof}

\section{Example}

Let us give a simple example of an equilibrium solution of a master
equation derived from a complex balanced equilibrium solution of the
rate equation. Consider the following stochastic Petri net with two
species:
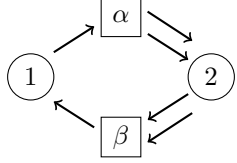
\[
\begin{tikzpicture}
\node [shape=circle,draw=black] at (-1.2,0) {$1$};
\node [shape=rectangle,minimum height=15pt,minimum width=15pt,inner sep=0pt,draw=black] at (0,0.8) {$\alpha$};
\node [shape=rectangle,minimum height=15pt,minimum width=15pt,inner sep=0pt,draw=black] at (0,-0.8) {$\beta$};
\node [shape=circle,draw=black] at ( 1.2,0) {$2$};
\draw [->, thick,shorten >=12pt, shorten <=12pt]  (-1.25,0.1) to (0,0.9) ; 
\draw [->, thick,shorten >=12pt, shorten <=12pt]  (0,0.8) to (1.25,0) ;
\draw [->, thick,shorten >=12pt, shorten <=12pt]  (0,1.1) to (1.3,0.25) ;
\draw [->, thick,shorten >=12pt, shorten <=12pt]  (1.25,0) to (0,-0.8);
\draw [->, thick,shorten >=12pt, shorten <=12pt]  (1.3,-0.25) to (0,-1.1); 
\draw [->, thick,shorten >=12pt, shorten <=12pt]  (0,-0.9) to (-1.25,-0.1) ; 
\end{tikzpicture}
\]
This represents a system where a single molecule of a diatomic gas
splits into two individual atoms with rate constant $\alpha$, while
two atoms combine to form a molecule with rate constant $\beta$.  The
rate equation is:
\begin{align*} 
\frac{d x_1}{d t} &= -\alpha x_1 + \beta x_2^2 \\ 
\frac{d x_2}{d t} &=  2\alpha x_1 - 2\beta x_2^2
\end{align*}
where $x_1$ and $x_2$ are the concentrations of species 1 and 2,
respectively. Equilibrium occurs precisely when $\alpha x_1 = \beta
x_2^2$.

The deficiency zero theorem implies that any equilibrium solution of
the rate equation for this stochastic Petri net is complex balanced
\cite{F,HJ}.  Thus, we can apply the Anderson--Cranciun--Kurtz theorem
to obtain equilibrium solutions of the master equation. The master
equation takes the form
\[
\frac{d}{d t} \Psi(t) = H \Psi (t)
\]
with
\[
H = \alpha (a_2^\dagger a_2^\dagger - a_1^\dagger) a_1 + 
\beta (a_1^\dagger - a_2^\dagger a_2^\dagger )a_2 a_2
= (z_2^2 - z_1) \, \left(\alpha \frac{\partial}{\partial z_1}
- \beta \frac{\partial^2}{\partial z_2^2}\right).
\]
Given the complexity of this Hamiltonian, it might seem difficult to
find exact solutions of $H \Psi = 0$.  But by the
Anderson--Craciun--Kurtz theorem, whenever $c$ gives a complex
balanced equilibrium for the rate equation, meaning $\alpha c_1 =
\beta c_2^2$, the following coherent state is a solution:
\[
\Psi_c = \frac{e^{c_1 z_1 + c_2 z_2}}{e^{c_1 + c_2}}.
\]
To see this, we need only note that
\[
(z_2^2 - z_1) \left(\alpha \frac{\partial}{\partial z_1} -
 \beta \frac{\partial^2}{\partial z_2^2}\right) e^{c_1 z_1 + c_2 z_2} = 
(z_2^2 - z_1) (\alpha c_1 - \beta c_2^2) e^{c_1 z_1 + c_2 z_2} .
\]
This vanishes precisely when $\alpha c_1 = \beta c_2^2$.

Note that
\[
\Psi_c = \frac{1}{e^{c_1} e^{c_2}} \sum_{(n_1,n_2) \in \N^2} 
\frac{c_1^{n_1}}{n_1!} \frac{c_2^{n_2}}{n_2!}z_1^{n_1}z_2^{n_2},
\]
so
\[
\psi_{n_1,n_2} = \frac{1}{e^{c_1}} \frac{c_1^{n_1}}{n_1!} \; 
\cdot \; \frac{1}{e^{c_2}} \frac{c_1^{n_2}}{n_2!}.
\]
This is a product of two independent Poisson distributions. It may
seem disturbing that in a reaction where one species turns into
another there are equilibrium solutions where the numbers of things of
each kind are independent random variables.  Intuitively, it might
seem that having more things of species 1 would force there to be
fewer of species 2.

In fact this intuition is misleading.  However, there are other
equilibrium solutions of the master equation that are not products of
independent probability distributions.  To see this, suppose we take
any of our equilibrium solutions $\Psi_c$ and construct a state
$\Psi^{(n)}$ by setting the probability $\psi^{(n)}_{n_1,n_2}$ equal
to $\lambda\psi_{n_1,n_2}$ when $2n_1 + n_2 = n$ and zero otherwise,
where $\lambda$ is the constant required for $\Psi^{(n)}$ to be
normalized.

It is readily checked that $\Psi^{(n)}$ is also an equilibrium
solution for our reversible chemical system.  The key observation is
that this particular stochastic Petri net has a `conserved quantity':
the total number of atoms, $2n_1 + n_2$.  A pure state $(n_1,n_2)$ may
only transition to another pure state $(n_1',n_2')$ if $2n_1+n_2 =
2n_1'+n_2'$.  This allows us to take any equilibrium solution of the
master equation and, in the language of quantum mechanics, `project
down to the subspace' where this conserved quantity takes a definite
value, to obtain a new equilibrium solution.  In the language of
probability theory, this is `conditioning on' the conserved quantity
taking a definite value.

Thanks to the stochastic version of Noether's theorem, conserved
 quantities correspond to symmetries of the Hamiltonian \cite{BF}.  In
 the example at hand, the conserved quantity is
\[   O = 2N_1 + N_2,   \]
where $N_i$ is the \define{number operator} for the $i$th species:
\[   N_i = a_i^\dagger a_i  .\]
The states $z^n = z_1^{n_1} z_2^{n_2}$ are a basis of eigenvectors of
these number operators:
\[   N_i z^n = n_i z^n,  \]
so they are also a basis of eigenvectors of $O$:
\[    O z^n = (2n_1 + n_2) z^n.  \]
By the argument already given, time evolution according to the
Hamiltonian $H$ preserves each eigenspace of $O$, so we have
\[   [\exp(t H), O] = 0  \]
for all $t \ge 0$, where the commutator $[A,B]$ of operators is
defined to be $AB - BA$.  This says that $O$ is a conserved quantity.

On the other hand, we also have
\[   [H ,\exp(s O)] = 0 \]
for all $s \in \R$.  This says that $O$ generates a 1-parameter group
of symmetries of the Hamiltonian $H$.  This symmetry can be seen to
have the following effect on coherent states:
\[   \exp(sO) \Psi_c = \frac{e^{c'}}{e^c} \Psi_{c'}  \]
where
\[           (c'_1, c'_2) = (e^{2s} c_1, e^s c_2)  .\]
Note that if $c$ is an equilibrium solution of the rate equation, so
is $c'$.

Modulo nuances of analysis involving unbounded operators, both $
[\exp(t H), O] = 0 $ and $[H , \exp(s O)] = 0$ are equivalent to the
more easily checked equation
\[    [H,O] = 0 ,\]
and both are also equivalent to
\[   [\exp(t H), \exp(s O)] = 0 \]
for all $s, t \in \R$.  This last equation says that applying the
symmetry $\exp(sO)$ to a state $\Psi$ and then evolving it for a time
$t$ according to the master equation gives the same result as first
evolving $\Psi$ for a time $t$ according to the master equation and
then applying the symmetry.

This example illustrates some general phenomena.  In chemical
reactions, the total number of atoms of each element is conserved, so
there will always be some conserved quantities $O_1, \dots, O_n$ that
are linear combinations of number operators.  Such conserved
quantities always commute, since number operators commute.  Given a
complex balanced classical state $c$ we can obtain an equilibrium
solution $\Psi_c$ of the master equation using the
Anderson--Craciun--Kurtz theorem.  We can then obtain other
equilibrium solutions of the master equation in two ways:

\begin{itemize}
\item We can project down $\Psi_c$ to any subspace where the
observables $O_i$ take definite values:
\[           \{ \Psi : O_i \Psi = \lambda_i \Psi  \} \]
for some nonnegative numbers $\lambda_i$.
\item We can apply symmetries of the form $\exp(s_1 O_1 + \cdots + s_n
O_n)$ to $\Psi_c$.  As well known in quantum mechanics, the result
will be, up to normalization, another coherent state:
\[       \exp(s_1 O_1 + \cdots + s_n O_n) \Psi_c = 
\frac{e^{c'}}{e^c} \Psi_{c'}  \]
\end{itemize}

\subsection*{Acknowledgements}
This work was done at the Centre for Quantum Technologies, Singapore,
which supported BF with a summer internship.


\begin{thebibliography}{100}

\bibitem{ACK} D.\ F.\ Anderson, G.\ Craciun, and T.\ G.\ Kurtz,
Product-form stationary distributions for deficiency zero chemical
reaction networks, {\sl Bull.\ Math.\ Bio.} {\bf 72} (2010) 1947--70. 
Also available as \href{http://arxiv.org/abs/0803.3042}
{arXiv:0803.3042}.

\bibitem{AK} D.\ F.\ Anderson and T.\ G.\ Kurtz, Continuous time
Markov chain models for chemical reaction networks, in {\sl Design and
Analysis of Biomolecular Circuits: Engineering Approaches to Systems
and Synthetic Biology}, H.\ Koeppl {\it et al.}, eds., Springer,
2011. Also available at
\href{http://www.math.wisc.edu/~anderson/papers/SURVEY_AndKurtz.pdf}
{http:/$\!$/www.math.wisc.edu/$\sim$anderson/papers/SURVEY$\underline{\;\;}$AndKurtz.pdf}.

\bibitem{A} R.\ Aris, Prolegomena to the rational analysis of systems
of chemical reactions, {\sl Arch.\ Rational Mech.\ Anal.}, {\bf 19}
(1965), 81--99.

\bibitem{B} J.\ Baez, Quantum techniques for reaction networks, 
available as \href{http://arxiv.org/abs/1306.3451} {arXiv:1306.3451}.


\bibitem{BB} J.\ Baez and J.\ Biamonte, A course on quantum techniques
for stochastic mechanics, available as
\href{http://arxiv.org/abs/1209.3632} {arXiv:1209.3632}.

\bibitem{BF} J.\ Baez and B.\ Fong, A Noether theorem for Markov
processes, {\sl J.\ Math. Phys.} {\bf 54} (2013), 013301.  Also available
as \href{http://arxiv.org/abs/1203.2035} {arXiv:1203.2035}.

\bibitem{D} M.\ Doi, Second quantization representation for classical
many-body systems, {\sl J.\ Phys.\ A} {\bf 9} (1976), 1465--1477.

\bibitem{F} M.\ Feinberg, Chemical reaction network structure and the
stability of complex isothermal reactors: I. The deficiency zero and
deficiency one theorems, {\sl Chemical Engineering Science} {\bf 42}
(1987), 2229--2268.  Also available at \hfill \break
\href{http://www.seas.upenn.edu/~jadbabai/ESE680/Fei87a.pdf}
{http://www.seas.upenn.edu/$\sim$jadbabai/ESE680/Fei87a.pdf}.

\bibitem{GP} P.~J.~E. Goss and J.\ Peccoud, Quantitative modeling of 
stochastic systems in molecular biology by using stochastic Petri nets,
{\sl Proc. Natl. Acad. Sci. USA} {\bf 95} (1998), 6750--6755.

\bibitem{G} J.\ Gunawardena, Chemical reaction network theory for
\textit{in-silico} biologists, lecture notes, available at
\href{http://vcp.med.harvard.edu/papers/crnt.pdf}
{http://vcp.med.harvard.edu/papers/crnt.pdf}.

\bibitem{GS} P.\ Grassberger and M.\ Scheunert, Fock-space methods for
identical classical objects, {\sl Fortsch.\ Phys.} {\bf 28} (1980),
547--578.

\bibitem{HJ} F.\ Horn and R.\ Jackson, General mass action kinetics,
{\sl Arch.\ Rat.\ Mech.\ Analysis}, {\bf 47}, 81--116, 1972.

\bibitem{K} I.\ Koch, Petri nets---a mathematical formalism to analyze 
chemical reaction networks, {\sl Molecular Informatics}, {\bf 29} (2010),
838--843.

\bibitem{KRS} I.\ Koch, W.\ Reisig and F.\ Schreiber, {\sl Modeling
in Systems Biology}, Springer, 2011.

\bibitem{KS} J.\ R.\ Klauder and B.\ Skagerstam, {\sl Coherent States}, 
World Scientific, Singapore, 1985.

\bibitem{PR} C.\ A.\ Petri and W.\ Reisig, Petri net, {\sl
Scholarpedia} {\bf 3} (2008), 6477.  Available at
\href{http://dx.doi.org/10.4249/scholarpedia.6477}
{doi:10.4249/scholarpedia.6477}.

\bibitem{Wil06} D.\ J.\ Wilkinson, {\sl Stochastic Modelling for
Systems Biology}, Taylor and Francis, New York, 2006.

\end{thebibliography}
\end{document}